\RequirePackage{fix-cm}
\documentclass[smallextended]{svjour3}       
\smartqed  
\usepackage[sort&compress,numbers]{natbib}
\usepackage{graphicx}
\usepackage{hyperref}
\hypersetup{
    colorlinks=true,
    linkcolor=magenta,
    filecolor=blue,      
    urlcolor=blue,
    citecolor=blue,
    pdftitle={Sharelatex Example},
    bookmarks=true,
    }
\usepackage{amsfonts}   
\newtheorem{mydef}{Definition}
\usepackage{amsmath}
\usepackage{graphicx}
\usepackage{tcolorbox}
\newtheorem{mytheo}{Theorem}
\usepackage{algorithm}
\usepackage{algcompatible}
\usepackage{url}
\begin{document}

\title{An Approximate Marginal Spread Computation Approach for the Budgeted Influence Maximization with Delay%
\thanks{Major part of this work was done when the first author was a PhD student at IIT Kharagpur. This work is supported by the following two grants E-Business Center of Excellence and Post Doctoral Fellowship Grant by IIT Gandhinagar.}
}


\author{Suman Banerjee        \and
        Mamata Jenamani \and Dilip Kumar Pratihar
}


\institute{ Suman Banerjee \at
              Department of Computer Science and Engineering,\\ IIT Gandhinagar. \email{suman.b@iitgn.ac.in}           
           \and
           Mamata Jenamani \at
             Department of Industrial and Systems Engineering,\\ IIT Kharagpur. \email{mj@iem.iitkgp.ac.in}
             \and
           Dilip Kumar Pratihar\\
            Department of Mechanical Engineering,\\ IIT Kharagpur. \email{dkpra@mech.iitkgp.ac.in}           
}

\date{Received: date / Accepted: date}

\maketitle

\begin{abstract}
Given a \emph{social network} of users with \emph{selection cost} and a fixed \textit{budget}, the problem of \textit{Budgeted Influence Maximization} finds a subset of the nodes ( known as \emph{seed nodes}) for initial activation to maximize the influence, such that the total selection cost is within the allocated budget. Existing solution methodologies for this problem make two  assumptions, which are not applicable to real\mbox{-}life situations. First, an influenced node of the current time stamp can trigger only once in the next time stamp to its inactive neighbors and the other one is the diffusion process continues forever. To make the problem more practical, in this paper, we introduce the \emph{Budgeted Influence Maximization with Delay} by relaxing the single time triggering constraint and imposing an additional constraint for maximum allowable diffusion time. For this purpose, we consider a delay distribution for each edge of the network, and consider a node is influenced, if it is so, within the allowable diffusion time. We first propose an incremental greedy strategy for solving this problem, which works based on the approximate computation of marginal gain in influence spread. Next, we make two subsequent improvements of this algorithm in terms of efficiency by exploiting the \textit{sub\mbox{-}modularity} property of the time delayed influence function. We implement the proposed methodologies with three benchmark datasets. Reported results show that the seed set selected by the proposed methodologies can lead to more number of influenced nodes compared to that obtained by other baseline methods. We also observe that between the two improvised methodologies, the second one is more efficient for the larger datasets.
\end{abstract}

\section{Introduction} \label{intro}
Consider the situation, when a commercial house wants to promote a brand among the customers through the on\mbox{-}line social networks. The easiest way to do so, is by initially selecting a set of influential users and distributing them free samples. If they like it, then they will influence their neighbors towards purchasing the item. At least some of them will purchase, and influence their neighbors. This cascading process will  continue and ultimately, a large fraction of the users will try the item. This helps the E\mbox{-}commerce house to earn more revenue. Here, the underlying problem is locating a set of highly influential users for initial activation to maximize the influence in the given social network. This problem is popularly referred to as the \textit{Social Influence Maximization Problem} (\textit{SIM Problem}), in the literature \cite{kempe2003maximizing, rui2020neighbour}. Due to its potential applications in several domains, such as \textit{personalized recommendation} \cite{song2006personalized}, \textit{viral marketing} \cite{tang2017profit}, trust inferencing \cite{pal2019trust} etc. this problem remains an active area of research in \textit{Computational Social Network Analysis} domain, since the last one and half decades \cite{banerjee2017survey}. 
\par  In reality, the social networks are formed by some rational and self\mbox{-}interested human agents. Hence, if a user is selected to be initially active, then incentivization is required. In this scenario, the SIM Problem is not realistic, as it assumes uniform incentive demand (selection cost) for all the users. However, in practice, it may be different for different users. By relaxing this assumption, Nguyen and Zheng \cite{nguyen2013budgeted} introduced the \emph{Budgeted Influence Maximization (BIM)} problem . For a social network of users with non\mbox{-}uniform selection cost and a fixed budget, this problem asks for selecting a set of initial nodes within the budget that leads to the maximum number of influenced nodes. There exist some solution methodologies for this problem such as approximation algorithm \cite{nguyen2013budgeted}, balanced seed selection approach \cite{han2014balanced}, community\mbox{-}based method \cite{banerjee2019combim} etc. However, in all these studies, it is implicitly assumed that, (i) an influenced node at current time stamp can trigger its inactive neighbors only in the next time stamp and, (ii) information can propagate forever. However, in reality, for many campaigns, diffusion time plays a vital role \cite{quan2018repost}. Here, we quote a few examples:
\begin{itemize}
	\item Consider the situation of viral marketing of a seasonal product. As these items are useful during a certain period, influencing a user towards purchasing an item of this category beyond that period will not be beneficial. 
	\item Consider the political campaigns before \textit{Prime Ministerial} or \textit{Presidential} etc. elections of a democratic country. Different political parties do campaign to influence the population to earn a majority opinion in their favor. From the situation itself, it is clear that, if a person is influenced towards a political party for voting after the date of the election, it does not make any sense.
	\item Imagine the situation of the viral marketing done by the organizers of a pop concert. They distribute free (or discounted) tickets among a few people, so that they can trigger a massive campaign and this may lead to a houseful show. However, if a person is influenced after the date of the show, it does not help the organizers to make the event successful.
\end{itemize} 
These real\mbox{-}life situations motivate us to study the problem of budgeted influence maximization by considering the propagation delay, i.e., counting influenced nodes only within the maximum allowable diffusion time. 
\par The main reason behind the social influence is the information  diffusion. To understand this process, there are many diffusion models that have been proposed and studied \cite{guille2013information}. Among them the  \emph{Independent Cascade Model (IC Model)} is quite popular and extensively used in influence maximization literature \cite{tang2014influence}. One unrealistic assumption of this model is that, every active node at time $t$ will trigger just once to activate its inactive neighbour. However, in reality, an influenced user can influence his uninfluenced neighbor by incurring some delay as well. That means, these models do not take care of the delay that happens in the real\mbox{-}world diffusion process. To cope up with this gap, recently the \textit{Latency Aware Independent Cascade Model} has been proposed by Liu et al. \cite{liu2014influence}. 
\par In this paper, we study the BIM Problem with delay  under the latency aware IC Model \cite{liu2014influence}. We propose an incremental greedy approach, which works based on approximate computation of marginal gain and two subsequent improvements by exploiting the sub\mbox{-}modularity property of the time delayed influence function. Particularly, we make the following contributions in this direction.
\begin{itemize}
	\item This paper studies the Budgeted Influence Maximization Problem under the Latency Aware Independent Cascade Model with an additional maximum allowable diffusion delay constraint.
	\item  We propose an \textit{approximate marginal gain computation} approach in influence spread and use this concept for solving this problem.
	\item Exploiting the sub\mbox{-}modularity property of the time delayed influence function, we make two subsequent improvements in terms of efficiency of the proposed algorithm.
	\item We implement the proposed methodologies, with three real\mbox{-}life publicly available social network datasets and perform a set of experiments for different budget values.  
\end{itemize}  
\par Rest of the paper is organized as follows: Section \ref{Sec:RW} states some relevant literature in and around of our study. Section \ref{sec:BPD} reports some background material and describes the problem formally. Section \ref{Sec:PM} describes the proposed solution methodologies for the BIM Problem with Delay. Section \ref{Sec:EE} contains the experimental evaluation of the proposed methodologies and finally, Section \ref{Sec:CFD} concludes this study and gives future directions.

\section{Related Work} \label{Sec:RW}
The primary basis of our study is the problem of social influence maximization and more specifically, influence maximization with diffusion delay. Here, we report some existing studies in and around this problem from the literature.
\paragraph{Influence Maximization and its Variants} As mentioned , the problem of influence maximization is all about to select a set of influential initial adopters such that the diffusion starting with them results into a significant influence \cite{kempe2003maximizing}. Though, the problem was initiated by Domingos and Richardson \cite{domingos2001mining}, Kempe et al. \cite{kempe2003maximizing} were the fist to investigate the computational issues of this problem and proposed a $(1-\frac{1}{e})$ factor approximation algorithm. This seminal work triggers a massive interest and significant amount of research has been carried out since last two decades. Please look into \cite{banerjee2017survey} (and references therein) for a comprehensive survey. Different kinds of solution methodologies have been prosed such as approximation algorithms \cite{tang2014influence}, heuristics solutions \cite{goyal2011simpath}, soft computing\mbox{-}based approachs \cite{tang2020discrete} and so on. Due to many practical applications, this problem has also been studied in different variants such as influence spectrum problem, target set selection problem, multi\mbox{-}round influence maximizaion problem and so on \cite{banerjee2017survey}.

\paragraph{Influence Maximization with Diffusion Delay} Most of the existing studies on influence maximization do not take care of the delay that happens in real\mbox{-}life diffusion process. Recently, there are few works considering diffusion time as a factor \cite{chen2012time, mohammadi2015time, shi2016retrieving}. Liu et al. \cite{liu2014influence} studied the \textit{time constrained influence maximization} problem and proposed influence spreading path\mbox{-}based solution approach for this problem. Chen et al. \cite{chen2012time} investigated the SIM Problem by considering the influenced nodes within a given deadline under the \textit{time delayed IC Model}. They proposed two heuristic solutions for this problem, where the first one is based on the dynamic programming technique for computing the influence spread in \textit{trees}, and the other one fits the problem in the original IC Model and applies the first heuristic. Mohammadi et al.  \cite{mohammadi2015time} studied the time sensitive influence maximization problem under their proposed Delayed IC and LT Models. They modified the existing incremental greedy solution for this problem and also proposed two centrality measures. Li et al. \cite{li2017dominated} studied the \textit{Dominated Competitive Influence Maximization Problem} for dealing with multiple kinds of information together.
\par In this paper, we study the budgeted influence maximization problem by considering the delay in the diffusion and also, enforce the constraint that a node is influenced if it so with the maximum allowable diffusion time. To the best of the authors' knowledge, this is the first study on the budgeted influence maximization problem, which considers both the diffusion delay as well as maximum diffusion time. In the next section, we state the required background knowledge and define our problem formally.

\section{Background and Problem Definition} \label{sec:BPD}
 The social network is given as a \emph{directed}, \emph{vertex and edge\mbox{-}weighted graph} $\mathcal{G}(V, E, \mathcal{P}, \theta)$, where $V(\mathcal{G})=\{u_1, u_2, \dots, u_n\}$ are the set of users, $E(\mathcal{G})=\{e_1, e_2, \dots, e_m\}$ are the relations among the users, and $E(\mathcal{G}) \subseteq V(\mathcal{G}) \times V(\mathcal{G})$. $\mathcal{P}$ and $\theta$ are the \textit{edge} and \textit{vertex weight functions}, which maps each edge and vertex to a number in between $0$ and $1$, i.e., $\mathcal{P}: E(\mathcal{G}) \longrightarrow (0,1]$ and $ \theta: V(\mathcal{G}) \longrightarrow [0,1]$, respectively. In the context of information diffusion, the weight of the edge $(u_iu_j)$ is considered as the \textit{diffusion probability from the user $u_i$ to $u_j$} and denoted by $\mathcal{P}_{u_i \rightarrow u_j}$. The vertex weight is treated as the \textit{diffusion threshold} (a measurement of how difficult to influence a user). The more the threshold value, it becomes harder to influence the user and for the user $u_i$ it is denoted as $\theta_i$. If $(u_iu_j) \notin E(\mathcal{G})$, then $\mathcal{P}_{u_i \rightarrow u_j}=0$. We denote the number of nodes and edges of $\mathcal{G}$ by $n$ and $m$, respectively.

\subsection{Independent Cascade Model with Delay} 
In independent Cascade Model (IC Model), information is diffused in discrete time steps from a set of initially active nodes, known as seed nodes. In IC Model, it is assumed that the seed nodes are active at time $t=0$. Each active (i.e., influenced) node $u$ at time $t$ will try to make every inactive (i.e., not influenced) neighbor (assume that $v$) active with probability $\mathcal{P}_{u \rightarrow v}$ and succeed if $\mathcal{P}_{u \rightarrow v} \geq \theta_{j}$. If this happens, then the user $v$ will be influenced at time $t+1$. Only the nodes that are active at current time stamp take part into the triggering process. A node can change its state from inactive to active, not the vice\mbox{-}versa. Once a node becomes active, it remains active forever. At the end of the diffusion process, the number of influenced nodes by the seed set $\mathcal{S}$ is captured by the social influence function and denoted as $\sigma(\mathcal{S})$. This is basically a set function defined on the ground set $V(\mathcal{G})$, which assigns each subset of $V(\mathcal{G})$ to a positive integer, i.e., $\sigma: 2^{V(\mathcal{G})} \rightarrow \mathbb{Z}^{+}$. Though the IC Model is popular for modeling the influence spread in a social network, it is not always practical, as it strictly enforces the single time triggering by a currently active node. However, in reality, during the diffusion, delay may exist, which is not captured by the IC Model. Considering this realistic phenomenon, recently, Latency\mbox{-}Aware Independent Cascade Model has been introduced by Liu et al. \cite{liu2014influence}. In this model $\forall (uv) \in E(\mathcal{G})$ along with the diffusion probability, a delay distribution $(\phi_1, \phi_2, \dots)$ is given. Suppose the edge $(uv)$ has the diffusion probability $\mathcal{P}_{u \rightarrow v}$ and a delay distribution $(\phi_1, \phi_2, \dots, \phi_l)$, where $\sum_{i=1}^{l} \phi_{i}=1$, $\phi_1 \geq \phi_2 \geq \dots \geq \phi_l$ and $\phi_{l+1}=0$. If a user $u_i$ becomes active at time stamp $t$, then he will try to activate $u_j$ with probability $\phi_1 \mathcal{P}_{u \rightarrow v}$ at $t+1$, with probability $\phi_2 \mathcal{P}_{u \rightarrow v}$ at $t+2$, and so on. After $t+l$, the diffusion probability from the $u$ to $v$ will be $0$. It is worthwhile to mention that we enforce the constraint of maximum diffusion time and consider a node to be influenced, if it is so, with the allowable time. So, in our example, $u_j$ is influenced, if $p \in \{1,2, \dots, l\}$ and $t+p \leq \mathcal{T}$. In this paper, we study the BIM problem with delay constraint under this diffusion model.

\subsection{Budgeted Influence Maximization with Diffusion Delay}
In BIM problem, there is a cost function, $\mathcal{C}:V(\mathcal{G}) \longrightarrow \mathbb{Z}^{+}$, which assigns each user its selection cost. For the user $u_i$, its selection cost is denoted by $\mathcal{C}(u_i)$. For a subset of users $\mathcal{S}$, its cost is defined as $\mathcal{C}(\mathcal{S})=\underset{u \in \mathcal{S}}{\sum}\mathcal{C}(u)$. The influence caused due to the seed set $\mathcal{S}$ is defined as the number of nodes influenced by them at the end of diffusion process and denoted by $\mathcal{I}(\mathcal{S})$. It is measured in terms of expectation and denoted as $\mathbb{E}[|\mathcal{I}(\mathcal{S})|]$. Hence, $\sigma(\mathcal{S})=\mathbb{E}[|\mathcal{I}(\mathcal{S})|]$, where $\sigma(.)$ is the social influence function returning the number of influenced nodes for a given seed set, i.e., $\sigma: 2^{V(\mathcal{G})} \rightarrow \mathbb{R}^{+}$. Now, the BIM problem  asks for selecting a seed set $\mathcal{S}$, which maximizes $\mathbb{E}[|\mathcal{I}(\mathcal{S})|]$ and also, the total selection cost should not exceed the budget, i.e., $\mathcal{C}(\mathcal{S}) \leq \mathcal{B}$. In this paper, we study the BIM problem with time delay constraint. We denote the number of influenced nodes within the time $\mathcal{T}$ due to the seed set $\mathcal{S}$ by $\sigma^{\mathcal{T}}(\mathcal{S})$, which is equal to $E[|\mathcal{I}^{\mathcal{T}}(\mathcal{S})|]$ and our goal is to maximize this quantity within the budget $\mathcal{B}$. Formally, the BIM problem with delay can be described as follows:

\begin{tcolorbox}

\underline{\textsc{ Budgeted Influence Maximization Problem with Delay}} \\
\textbf{Input:}Social Network $\mathcal{G}(V, E, \mathcal{P}, \theta)$, a cost function $\mathcal{C}$, Budget $\mathcal{B} $, a latency distribution $\mathcal{P}^{\mathcal{L}}$ and Maximum Delay $\mathcal{T}$.

\textbf{Problem:} Find out the seed set ($\mathcal{S}$) such that $\underset{u \in \mathcal{S}}{\sum} \mathcal{C}(u) \leq \mathcal{B}$ and for any other seed set $\mathcal{S}^{'}$ with $\underset{v \in \mathcal{S}^{'}}{\sum} \mathcal{C}(v) \leq \mathcal{B}$, $ E[|\mathcal{I}^{\mathcal{T}}(\mathcal{S})|] \geq E[|\mathcal{I}^{\mathcal{T}}(\mathcal{S}^{'})|]$).
\end{tcolorbox}

\par In the next section, we describe the proposed methodologies for solving this problem.
\section{Proposed Methodology} \label{Sec:PM}
This section is broadly divided into three subsection. In the first one, we state the intuitive solution approach based on the solution methodology proposed by Nguyen and Zheng \cite{nguyen2013budgeted} for the BIM Problem. Next, we describe the proposed methodology based on the approximate computation of marginal influence spread. Finally, the last subsection contains two algorithms, which improves the efficiency of the proposed methodology in terms of computational time by exploiting the sub\mbox{-}modularity property of the time delayed influence function. However, prior to that, we present two preliminary definitions and their modifications, which will be used in our proposed methodologies.   
\begin{mydef}[Marginal Influence Gain] \label{Def:1}
	For a given seed set $\mathcal{S}$ and a node $u \in V(\mathcal{G}) \setminus \mathcal{S}$, the marginal influence gain for the node $u$ with respect to the seed set $\mathcal{S}$ is defined as the difference in the number of influenced nodes, when the seed sets are $\mathcal{S} \cup \{u\}$ and $\mathcal{S}$, respectively, and it is denoted as $\Delta_{\mathcal{I}}(\mathcal{S} \vert u)$. Hence,
	\begin{equation}
	\Delta_{\mathcal{I}}(\mathcal{S} \vert u)= \mathbb{E}[|\mathcal{I}(\mathcal{S}\cup \{u\})|] - \mathbb{E}[|\mathcal{I}(\mathcal{S})|]
	\end{equation}  
\end{mydef}
As we are counting the number of influenced nodes within the maximum allowable diffusion time $\mathcal{T}$, we modify the Definition \ref{Def:1} by imposing this constraint and define the \textit{marginal influence gain within the allowable diffusion time} as follows:
\begin{equation}
\Delta^{\mathcal{T}}_{\mathcal{I}}(\mathcal{S} \vert u)= \mathbb{E}[|\mathcal{I}^{\mathcal{T}}(\mathcal{S}\cup \{u\})|] - \mathbb{E}[|\mathcal{I}^{\mathcal{T}}(\mathcal{S})|]
\end{equation} 
Next, we define an important property of a set function.
\begin{mydef}[Submodularity of a Set Function]
	A set function $f(.)$ defined on the ground set $V(\mathcal{G})$ is submodular, if $\forall \mathcal{S}_{1} \subseteq \mathcal{S}_{2} \subset V(\mathcal{G})$ and $u \in V(\mathcal{G}) \setminus \mathcal{S}_{2}$, the following condition always holds
	\begin{center}
		$f(\mathcal{S}_{1} \cup \{u\})-f(\mathcal{S}_{1}) \geq f(\mathcal{S}_{2} \cup \{u\})-f(\mathcal{S}_{2})$
	\end{center}
\end{mydef}
Kempe et al. \cite{kempe2003maximizing} showed that the social influence function is sub\mbox{-}modular under IC model. In \cite{liu2014influence}, authors showed that the social influence function is sub\mbox{-}modular under latency aware independent cascade model, as well. We use these results in the proposed methodologies.
\subsection{Intuitive Solution Approach}
Intuitively, the easiest approach to solve the problem is the incremental greedy strategy, as presented for the SIM Problem in \cite{kempe2003maximizing} and for the BIM problem in \cite{nguyen2013budgeted}. Starting with the empty seed set $\mathcal{S}$, this process iteratively adds a node that makes the maximum marginal gain per unit cost in expected influence spread within the given time. Hence, if $\mathcal{S}^{i}$ is the seed set and $\mathcal{B}^{i}$ is the remaining budget after the $i^{th}$ iteration, the node $u$ will be added to the set $\mathcal{S}$ in $(i+1)^{th}$ iteration, i.e., $\mathcal{S}^{i+1}= \mathcal{S}^{i} \cup \{u\}$, if the condition in Equation (\ref{Eq:1}) is met.
\begin{equation} \label{Eq:1}
u=\underset{v \in V(\mathcal{G}) \setminus \mathcal{S}^{i}, \mathcal{C}(v)\leq \mathcal{B}^{i}}{argmax} \frac{\Delta^{\mathcal{T}}_{\mathcal{I}}(\mathcal{S} \vert v)}{\mathcal{C}(v)}
\end{equation} 

Algorithm \ref{Algo:Intutive} describes this procedure.
\begin{algorithm} [H]
	\caption{Proposed Methodology for the Budgeted Influence Maximization Problem with Time Delay Problem.}
	\label{Algo:Intutive}
	\begin{algorithmic}[1]
		
		\renewcommand{\algorithmicrequire}{\textbf{Input:}}
		
		\renewcommand{\algorithmicensure}{\textbf{Output:} }
		
		\REQUIRE $\mathcal{G}(V, E, \mathcal{P}, \theta)$,
		$\mathcal{C}: V(\mathcal{G}) \longrightarrow \mathbb{Z}^{+}$, Delay Distribution $\mathcal{P}^{\mathcal{L}}$ and  $\mathcal{B}$. 
		
		
		\ENSURE Seed Set $\mathcal{S} \subseteq V(\mathcal{G})$ with $\mathcal{C}(\mathcal{S}) \leq \mathcal{B}$.
		\STATE $\mathcal{S} \longleftarrow \phi$\;
		
		\WHILE{$\mathcal{B}>0$}
		\STATE $u \longleftarrow \underset{v \in V(\mathcal{G}) \setminus \mathcal{S}, \mathcal{C}(\{v\}) \leq \mathcal{B}}{argmax}  \frac{\Delta^{\mathcal{T}}_{\mathcal{I}}(\mathcal{S} \vert v)}{\mathcal{C}(v)}$\;
		\IF{$u == \textbf{null}$}
		\STATE $break$\;
		\ENDIF 
		\STATE $\mathcal{S} \longleftarrow \mathcal{S} \cup \{u\}$\;
		\STATE $\mathcal{B} \longleftarrow \mathcal{B} - \mathcal{C}(u)$\;
		\ENDWHILE
		\STATE $return \ \mathcal{S}$\;
	\end{algorithmic}
\end{algorithm}

Now, it is easy to observe that the important component of Algorithm \ref{Algo:Intutive} is to compute the expected influence spread with delay for a given seed set. The straight forward way to compute this is the following. \\
\textbf{Influenc Estimation Procedure:} Every node of the network stores its activation time and current status, which can be $influenced$ or $influenced \ with \ delay$ or $not \ influenced$. Initially, we set the status of all the seed nodes as $influenced$ and their activation time as $0$. Then, we keep on iterating for successive time instance, $t=1,2,3, \dots$ until there are no more nodes with status $influenced$ or $influenced \ with \ delay$ in the current time instance. For any arbitrary time instance $t$, we consider each node $u$, that has been influenced in the $(t-1)$\mbox{-}th time stamp and we find out the neighbors of $u$, whose statuses are $not \ influenced$ and which can be influenced by $u$. A neighbor $v$ can be influenced by $u$, if the edge probability of $(u,v)$, i.e., $\mathcal{P}_{u \rightarrow v} \ast p_{t}$, is either greater than or equal to $\theta_{v}$ and $t$ is within $\mathcal{T}$. If the current status of $v$ is $not \ influenced$ and the condition is met, then its status is changed to $influenced \ with \ delay$ and its activation time is set to $t$. If the current status of $v$ is $influenced \ with \ delay$ and $t$ is less than the current activation time of $v$, then we update its activation time by $t$. Then, we change the status of all the nodes that can be activated in time instance, $t-1$ to $influenced$ and continue for the next time instance. At the end, we return the number of nodes with the status $influenced$. However, as mentioned in \cite{kempe2003maximizing}, this process is repeated $\mathcal{R}$ times and the average is returned as the approximate value of the $\mathbb{E}[|\mathcal{I}^{\mathcal{T}}(\mathcal{S})|]$. If, $|\mathcal{S}|=k$, then traversing $\mathcal{G}$ for $\mathcal{R}$ times from $k$ different nodes 
 requires $\mathcal{O}(k(m+n)\mathcal{R})$, which is also the time complexity of this influence estimation process.
\par Though the Algorithm \ref{Algo:Intutive} is simple and intuitive, as reported in the literature \cite{kempe2003maximizing}, \cite{nguyen2013budgeted}, it cannot be used for finding the seed set even in a network of size $1000$ nodes and edges. The main reason behind this, is that the influence estimation procedure as described is heavily time consuming. Hence, our main focus is to reduce the computational burden of the marginal influence gain by an approximate computation.

\subsection{Approximate Marginal Gain Computation\mbox{-}Based Approach} 
\par For any arbitrary node $u \notin \mathcal{S}$, the probability that it will be immediately influenced in the next time stamp by a given seed set $\mathcal{S}$ can be presented by Equation (\ref{Eq:2}).
\begin{equation} \label{Eq:2}
\mathcal{P}_{\mathcal{S}\rightarrow u} =
\begin{cases}
1- \underset{v \in \mathcal{S}, (vu) \in E(\mathcal{G})}{\prod}(1 - \mathcal{P}_{v \rightarrow u}) & \text{if $u \in \mathcal{N}^{out}({\mathcal{S})}$}, \\
0 & \text{otherwise}, 
\end{cases}
\end{equation}
where $\mathcal{N}^{out}(\mathcal{S})$ is defined as the set of nodes having a directed edge from at least one of the nodes in $\mathcal{S}$, i.e., $\mathcal{N}^{out}(\mathcal{S})=\{u \vert \exists u^{'} \in \mathcal{S} \text{ and } (u^{'}u) \in E(\mathcal{G}) \}$. From the sub\mbox{-}modularity property of $\mathcal{I}^{\mathcal{T}}(.)$, we have $ E[|\mathcal{I}^{\mathcal{T}}(\{v\})|] \geq E[|\mathcal{I}^{\mathcal{T}}(\mathcal{S} \cup \{v\})|] - E[|\mathcal{I}^{\mathcal{T}}(\mathcal{S})|]$, for all $\mathcal{S} \subset V(\mathcal{G})$ and $v \in V(\mathcal{G}) \setminus \mathcal{S}$. Hence, $E[|\mathcal{I}^{\mathcal{T}}(\{v\})|]$ can be multiplied by a suitable fraction, such that in each iteration, the marginal gain in influence spread can be computed efficiently. The fraction, which is to be multiplied, should follow the criteria that, if the marginal gain of a node with respect to the current seed set is more, then the value of that should also be more for this node and vice\mbox{-}versa. Now, we present our approximation strategy of marginal gain computation in influence spread of Theorem \ref{Th:1}.
\begin{mytheo} \label{Th:1}
	The approximate value of marginal gain in influence spread for the node $v$ and the seed set $\mathcal{S}^{i}$ can be given by Equation (\ref{Eq:3}).
	\begin{equation} \label{Eq:3}
	\mathbb{E}[|\mathcal{I}^{\mathcal{T}}(\mathcal{S}^{i} \cup \{v\})|] - \mathbb{E}[|\mathcal{I}^{\mathcal{T}}(\mathcal{S}^{i})|]\approx \mathbb{E}[|\mathcal{I}^{\mathcal{T}}(v)|] \frac{\underset{(vu) \in E(\mathcal{G})}{\sum}\mathcal{P}_{v \rightarrow u} (1-\mathcal{P}_{\mathcal{S} \rightarrow u})\mathbb{E}[|\mathcal{I}^{\mathcal{T}}(u)|]}{\underset{(vu) \in E(\mathcal{G})}{\sum}\mathcal{P}_{v \rightarrow u}\mathbb{E}[|\mathcal{I}^{\mathcal{T}}(u)|]}
	\end{equation}
\end{mytheo} 
\begin{proof}
	As stated previously, due to the sub\mbox{-}modularity property, the marginal gain in influence spread due to the user $v$ for the seed $\mathcal{S}$ will always be less than $\mathbb{E}[|\mathcal{I}^{\mathcal{T}}(v)|]$. Now, if $(vu) \in E(\mathcal{G})$, the quantity $\mathcal{P}_{v \rightarrow u} (1-\mathcal{P}_{\mathcal{S} \rightarrow u})$ represents the probability that none of the seed nodes, however, only $v$ can influence $u$. Multiplying this quantity with $\mathbb{E}[|\mathcal{I}^{\mathcal{T}}(u)|] $ gives the number of influenced nodes $v$ can generate by influencing the node $u$. If we sum up this quantities for all the neighbors of $v$, we obtain the number of influenced nodes due to the node $v$ only. Hence, the numerator of the fraction in the right hand size of Equation (\ref{Eq:3}) gives this quantity. It is important to observe, when the value of $\mathcal{P}_{\mathcal{S} \rightarrow u}$ is less, the value of the numerator will be more. The denominator gives the number of influenced nodes by the node $v$. Hence, the value of this fraction will always be less than $1$. This implies that the marginal gain will be less than $\mathbb{E}[|\mathcal{I}^{\mathcal{T}}(v)|]$. Also, the gain will be more, when the number of influenced nodes is due to the node $v$ only (not by any other nodes of $\mathcal{S}$). Hence, the quantity in the right hand side gives a suitable approximation of the marginal gain computation of influence spread.    
\end{proof}  
Now, if we compute the marginal gain in influence spread as described in Theorem \ref{Th:1}, we have the following advantage. According to Equation (\ref{Eq:3}), for the marginal gain computation of the node $v$, we need to estimate the spread by the node $v$ and its outgoing neighbors individually. On the contrary, if we go by the intuitive approach as described previously, in each iteration, we need to calculate the marginal gain for all the individual nodes not present in the current seed set with respect to the current seed set, which is heavily time consuming. So, if we compute the marginal gain in influence spread, as stated in Theorem \ref{Th:1}, computationally, it will be much more efficient. Algorithm \ref{Algo:2} describes the proposed methodology. 
\begin{algorithm} [h]
	\caption{Proposed Methodology for the Budgeted Influence Maximization Problem with Time Delay Problem.}
	\label{Algo:2}
	\begin{algorithmic}[1]
		
		\renewcommand{\algorithmicrequire}{\textbf{Input:}}
		
		\renewcommand{\algorithmicensure}{\textbf{Output:} }
		
		\REQUIRE  $\mathcal{G}(V, E, \mathcal{P}, \theta)$,
		$\mathcal{C}: V(\mathcal{G}) \longrightarrow \mathbb{Z}^{+}$, $\mathcal{P}^{\mathcal{L}}$, and $\mathcal{B}$. 
		
		

		\ENSURE Seed Set $\mathcal{S} \subseteq V(\mathcal{G})$ with $\mathcal{C}(\mathcal{S}) \leq \mathcal{B}$.
		

		\STATE $\mathcal{S} \longleftarrow \phi$\;
		
		\STATE $Flag = 1$\;
		
		\FOR{All $u \in V({\mathcal{G}})$}
		
		\STATE $\text{Compute } \frac{\sigma_{\mathcal{T}}(u)}{\mathcal{C}(u)} \text{using Equation \ref{Eq:3}}$\;
		
		\ENDFOR
		
		\STATE $v\longleftarrow \underset{u \in V({\mathcal{G}}), \mathcal{C}(u) \leq \mathcal{B}}{argmax} \frac{\sigma_{\mathcal{T}}(u)}{\mathcal{C}(u)}$\;
		
		\STATE $\mathcal{S} \longleftarrow \mathcal{S} \cup \{v\}$\;
		
		\FOR {All $u \in V({\mathcal{G}}) \setminus \{v\}$}
		
		\IF{$u \in \mathcal{N}^{out}(v)$}
		
		\STATE $\mathcal{P}_{\mathcal{S} \rightarrow u}= 1- \underset{w \in \mathcal{S}, (wu)\in E(\mathcal{G})}{\prod} (1- \mathcal{P}_{w \rightarrow u})$\;

		\ELSE
		
		\STATE $\mathcal{P}_{\mathcal{S} \rightarrow u}=0$\;
		
		\ENDIF
		
		\ENDFOR
		
		\WHILE{(Flag)}
		
		\STATE $Flag \leftarrow 0$\;
		
		\STATE $u \longleftarrow \underset{v \in V(\mathcal{G}) \setminus \mathcal{S}, \mathcal{C}(\mathcal{S} \cup \{v\}) \leq \mathcal{B}}{argmax} \frac{\sigma_{\mathcal{T}}(v)}{\mathcal{C}(v)} \frac{\underset{(vu) \in E(\mathcal{G})}{\sum}\mathcal{P}_{v \rightarrow u} (1-\mathcal{P}_{\mathcal{S} \rightarrow u})\sigma_{\mathcal{T}}(u)}{\underset{(vu) \in E(\mathcal{G})}{\sum}\mathcal{P}_{v \rightarrow u}\sigma_{\mathcal{T}}(u)}$\;
		
		\IF{$u == \textbf{null}$}
		
		\STATE $break$\;
		
		\ENDIF 
		
		\STATE $\mathcal{S} \longleftarrow \mathcal{S} \cup \{u\}$\;
		
		\STATE $Flag=1$\;
		
		
		\FOR {All $u \in \mathcal{N}^{out}(\mathcal{S})$}
		
		\STATE $\text{Update } \mathcal{P}_{\mathcal{S}\rightarrow u} \text{ using Equation \ref{Eq:2}}$\;

		\ENDFOR
		
		\ENDWHILE

		\STATE $return \ \mathcal{S}$\;
		
		
	\end{algorithmic}
\end{algorithm}
\par Algorithm \ref{Algo:2} illustrates the proposed methodology for solving the BIM Problem with delay. For a given social network $\mathcal{G}$ with selection cost of each user, a fixed budget $\mathcal{B}$ and maximum allowable diffusion time $\mathcal{T}$, Algorithm \ref{Algo:2} selects a reasonably well seed set, $\mathcal{S}$  within affordable computational time. The basic intuition behind the Algorithm \ref{Algo:2} is that a node is added to the seed set, if it has the maximum marginal influence gain (computed using Equation (\ref{Eq:3})) among all the nodes not already present in the seed set and if the cost of adding that node to the seed set is within the given budget. Initially, we make the seed set, $\mathcal{S}$ empty and the boolean variable \textit{Flag} to $1$.  Initially, we calculate the influence spread $\sigma_{\mathcal{T}}$ per unit cost due to each node in the graph, using the influence estimation process as described, and the node causing the maximum spread is included in the seed set $\mathcal{S}$. Then, for all the remaining nodes, we calculate $\mathcal{P}_{\mathcal{S} \rightarrow u}$ using Equation (\ref{Eq:2}). Next, we iteratively keep on adding a node to the seed set, $\mathcal{S}$ until the budget is exhausted or we cannot find another node within the given budget. In each iteration, we select a node with the largest marginal gain in influence spread, calculated using Equation (\ref{Eq:3}). Then, we update $\mathcal{P}_{\mathcal{S} \rightarrow u}$ for all nodes that are neighbors of the currently selected seed nodes. At the end of the Algorithm \ref{Algo:2} returns one of best quality seed set for diffusion.
\par Now, we investigate the computational time and space requirement of Algorithm \ref{Algo:2}. Lines $1$ and $2$ both will take $\mathcal{O}(1)$ time. Next, each individual node's influence spread is calculated using the influence estimation procedure, and among them, the node that causes the maximum spread is included in the empty seed set. Hence, Lines $3$ to $7$ lead to the time requirement of $\mathcal{O}(n(n+m)\mathcal{R})$. By using Line $8$ through $14$ diffusion probabilities of all the neighbors of the currently selected seed node are computed, which requires $\mathcal{O}(d^{out}_{m})$ time in the worst case, where $d^{out}_{m}$ is the maximum out degree of any node in $\mathcal{G}$. Now, the while loop starting from Lines $15$ through $26$ iteratively select the seed nodes by approximately computing the marginal gain in influence spread. Inside the loop, two operations are performed. The first one is to find the appropriate node at Line $17$ and this requires $\mathcal{O}(n)$ calls to the influence estimation procedure, in the worst case. However, we need not compute each node's individual influence spread value, as we have already computed once in Line $4$. Hence, the required time for executing Line $17$ will be of $\mathcal{O}(n d_{m}^{out})$. The second one is the updation of the influence probabilities of the neighboring nodes of the current seed nodes in Line $24$. Let $\mathcal{C}_{min}$ denotes the minimum selection cost among all the nodes. Cardinality of the seed set can be at most $\frac{\mathcal{B}}{\mathcal{C}_{min}}$ and  hence, the number of times the While loop runs will be at most $\frac{\mathcal{B}}{\mathcal{C}_{min}}$. For each iteration, the running time from Line 24 to 26 will be $\mathcal{O}(\frac{\mathcal{B}}{\mathcal{C}_{min}} n )$, where $d_{m}^{in}$ is the maximum in degree of any node in $\mathcal{G}$. Other than these two, rest of the steps lying inside the loop has the running time of $\mathcal{O}(1)$. Hence, for each iteration, the running time of the While loop starting from Line $15$ through $27$ will be  $\mathcal{O}(\frac{\mathcal{B}}{\mathcal{C}_{min}} n + n d_{m}^{out}) \simeq n(\frac{\mathcal{B}}{\mathcal{C}_{min}} + d_{m}^{out})$. As the number of iterations of the While loop is of $\mathcal{O}(\frac{\mathcal{B}}{\mathcal{C}_{min}})$, the total running time of the While loop is $\mathcal{O}(\frac{\mathcal{B}}{\mathcal{C}_{min}} n (\frac{\mathcal{B}}{\mathcal{C}_{min}} + d_{m}^{out}))$. Hence, the total running time of the Algorithm \ref{Algo:2} is $\mathcal{O}(n(n+m)\mathcal{R}+ d_{m}^{out} + \frac{\mathcal{B}}{\mathcal{C}_{min}} n (\frac{\mathcal{B}}{\mathcal{C}_{min}} + d_{m}^{out}))$. The extra space required by Algorithm \ref{Algo:2} is $\mathcal{O}(n)$ for storing the individual node's influence ability and $\mathcal{O}(n)$ for storing each node's influence probability from the seed set. The formal statement is presented as Theorem \ref{Th:2}.
\begin{mytheo} \label{Th:2}
	Running time and space requirement of Algorithm \ref{Algo:2} is of $\mathcal{O}(n(n+m)\mathcal{R}+ d_{m}^{out} + \frac{\mathcal{B}}{\mathcal{C}_{min}} n (\frac{\mathcal{B}}{\mathcal{C}_{min}} + d_{m}^{out}))$ and $\mathcal{O}(n)$, respectively.
\end{mytheo}
\subsection{Improving Efficiency of Algorithm \ref{Algo:2} by Exploiting Sub\mbox{-}modularity}
Though our proposed algorithm mitigates the scalability issue of Algorithm \ref{Algo:Intutive} to certain extent, it is not sufficient for processing of networks with the larger size. Therefore, we improvise Algorithm \ref{Algo:2} by reducing redundant marginal influence gain computation. Algorithm \ref{Algo:3} describes this procedure.
\begin{algorithm}[h] 
	\caption{Method for improvising Algorithm \ref{Algo:2} by exploiting the sub\mbox{-}modularity property of the time delayed influence function}
	\label{Algo:3}
	\begin{algorithmic}[1]
		
		\renewcommand{\algorithmicrequire}{\textbf{Input:}}
		
		\renewcommand{\algorithmicensure}{\textbf{Output:} }
		
		\REQUIRE $\mathcal{G}(V, E, \mathcal{P}, \theta)$, $\mathcal{C}: V(\mathcal{G}) \longrightarrow \mathbb{Z}^{+}$, $\mathcal{P}^{\mathcal{L}}$, $\mathcal{T}$ and  $\mathcal{B}$.

		\ENSURE Seed Set $\mathcal{S} \subseteq V(\mathcal{G})$ with $\mathcal{C}(\mathcal{S}) \leq \mathcal{B}$.
		

		\FOR {All $u \in V(\mathcal{G})$}

		\STATE $\text{Calculate } \sigma_{\mathcal{T}}(u)$\;
		
		\ENDFOR
		
		\STATE $u\longleftarrow\underset{v \in V(\mathcal{G}),\mathcal{C}(v)\le \mathcal{B}}{argmax}\sigma_{\mathcal{T}}(v)/\mathcal{C}(v)$\;
		
		\STATE $\mathcal{S}\leftarrow \{u\}$\;
		
		\STATE $\forall u \in V(\mathcal{G}),\delta_{\mathcal{S}}\leftarrow +\infty$\;
		
		\IF {$w \in \mathcal{N}^{out}(u)$}

		\STATE $\mathcal{P}_{\mathcal{S} \rightarrow w}\leftarrow \mathcal{P}_{u \rightarrow w}$\;

		\ELSE
		
		\STATE $\mathcal{P}_{\mathcal{S} \rightarrow w}\leftarrow 0$;
		
		\ENDIF
		
		\STATE $Flag = 1$\;
		
		\WHILE {$(Flag)$}

		\STATE $Flag \leftarrow 0$\;
		
		\FOR {All $u \in V(\mathcal{G})/\mathcal{S}$}

		\STATE $cur_{u} \leftarrow False$\;
		
		\STATE $\delta_{u} \leftarrow 0$\;
		
		\ENDFOR
		
		\WHILE {($True$)}
		
		\STATE $u \longleftarrow \underset{v \in V(\mathcal{G})/\mathcal{S},\mathcal{C}(\mathcal{S}\cup \{v\})\leq \mathcal{B}}{argmax} \delta_{u}$\;
		
		\IF {$(u = \textbf{null})$}

		\STATE $break$\;
		
		\ENDIF
		
		\IF {$(cur_{u})$}

		\STATE $\mathcal{S} \leftarrow \mathcal{S} \cup \{u\}$\;
		
		\STATE $Flag \leftarrow 1$\;
		
		\STATE $break$\;
		
		
		\ELSE
		
		\STATE $\delta_{u}\leftarrow \frac{\sigma_{\mathcal{T}}(\{u\})}{\mathcal{C}(u)} \frac{\underset{(uw) \in E(\mathcal{G})}{\sum}\mathcal{P}_{u \rightarrow w} (1-\mathcal{P}_{\mathcal{S} \rightarrow w}) \sigma_{\mathcal{T}}(\{w\})}{\underset{(uw) \in E(\mathcal{G})}{\sum}\mathcal{P}_{u \rightarrow w}\sigma_{\mathcal{T}}(\{w\})}$\;
		
		\STATE $cur_{u}\leftarrow True$\;
		
		\ENDIF
		
		\ENDWHILE
		
		\FOR {All $w \in \mathcal{N}^{out}(\mathcal{S})$}

		\STATE $\text{Update } \mathcal{P}_{\mathcal{S}\rightarrow w}$\;
		
		\ENDFOR
		
		\ENDWHILE
		
		\STATE $return \ \mathcal{S}$\;
		
		
	\end{algorithmic}
\end{algorithm}
\par Algorithm \ref{Algo:3} suggests an improvisation to our proposed methodology. The main bottleneck of Algorithm \ref{Algo:2} is that in each iteration, for all the non\mbox{-}seed nodes, we need to compute the marginal influence gain. However, it is important to observe that many of these computations are redundant and hence, can be avoided for improving the efficiency of Algorithm \ref{Algo:2} to a great extent. By the sub\mbox{-}modularity property of the time\mbox{-}delayed influence function, marginal influence gain of a non\mbox{-}seed node (say $u$) with respect to the seed set at $i$\mbox{-}th iteration ($\mathcal{S}^{i}$) is always either more than or equal to that with respect to seed set at $(i+1)$\mbox{-}th iteration. If we compute the marginal influence gain of the non\mbox{-}seed nodes, taken in descending order of their previously calculated marginal influence gains, then the sub\mbox{-}modularity can be exploited. Initially, as the seed set is empty, the individual influence spread is computed in the first iteration, and the node with the highest spread is included in the seed set. For the successive iterations, we compute the marginal influence gain of the non seed nodes in the descending order of these values in the previous iteration. In an arbitrary iteration, at some point, before we calculate the value for all the nodes, we may get a node for which the marginal influence gain has already been calculated in the present iteration. This is marked by a boolean array. Thus, we add that node to the seed set and stop the present iteration and are not calculating the marginal influence gain of all the non\mbox{-}seed nodes. Thus, we can reduce redundant marginal influence spread computations.
Though the asymptotic complexity remains the same, it is evident from our experimental results that Algorithm \ref{Algo:3} takes less computation time than that of Algorithm \ref{Algo:2}, when it is run with real-life social networks, when the dataset is quite large. In the literature, the exploitation of sub\mbox{-}modularity property of the influence function has been previously used by  \cite{leskovec2007cost} to improve the incremental greedy algorithm for SIM problem proposed by \cite{kempe2003maximizing}.
\begin{algorithm}[H]
	\caption{Method for improvising Algorithm \ref{Algo:3} by exploiting the sub\mbox{-}modularity property of the time delayed influence function}
	\label{Algo:4}
	\begin{algorithmic}[1]
		\renewcommand{\algorithmicrequire}{\textbf{Input:}}
		\renewcommand{\algorithmicensure}{\textbf{Output:} }
		\REQUIRE $\mathcal{G}(V, E, \mathcal{P}, \theta)$,
		$\mathcal{C}: V(\mathcal{G}) \longrightarrow \mathbb{Z}^{+}$, $\mathcal{P}^{\mathcal{L}}$, $\mathcal{T}$ and  $\mathcal{B}$. 
		
		\ENSURE Seed Set $\mathcal{S} \subseteq V(\mathcal{G})$ with $\mathcal{C}(\mathcal{S}) \leq \mathcal{B}$
		\STATE $\mathcal{S}\leftarrow \phi$\;
		\STATE $last\_seed\leftarrow null$\;
		\STATE $cur\_best \leftarrow null$\;
		\STATE $mg1 = Create\_Vector(len (V(\mathcal{G})),0)$\;
		\STATE $prev\_best = Create\_Vector(len (V(\mathcal{G})),null)$\;
		\STATE $mg2 = Create\_Vector(len (V(\mathcal{G})),0)$\;
		\STATE $flag = Create\_Vector(len (V(\mathcal{G})),0)$\;
		\FOR {$All \ u \in V(\mathcal{G})$}
		
		\STATE $mg1[u]\leftarrow\sigma_{\mathcal{T}}(u)$\;
		\STATE $prev\_best[u]\leftarrow cur\_best$\;
		\STATE $mg2[u]\leftarrow \sigma_{\mathcal{T}}(\{u,cur\_best\})$\;
		\STATE $flag[u]\leftarrow 0$\;
		\STATE $cur\_best\leftarrow\underset{u \in V(\mathcal{G}),\mathcal{C}(u)\le \mathcal{B}}{argmax} \ mg1[u]$\;
		\ENDFOR
		\algstore{myalg2}
	\end{algorithmic}
\end{algorithm}
\begin{algorithm}[h]
	\begin{algorithmic}[1]
		\algrestore{myalg2}
		\STATE $\mathcal{S}\leftarrow\mathcal{S}\cup\{cur\_best\}$\;
		\STATE $last\_seed\leftarrow cur\_best$\;
		\WHILE {($True$)}
		\STATE $u \longleftarrow \underset{v \in V(\mathcal{G})/\mathcal{S},\mathcal{C}(\mathcal{S}\cup\{v\})\le\mathcal{B}}{argmax}mg1[v]$\;

		\IF{$u == \textbf{null}$}
		
		\STATE $break$\;
		\ENDIF
		
		\IF{$(flag[u] == \mid\mathcal{S}\mid)$}
		
		\STATE $\mathcal{S}\leftarrow\mathcal{S}\cup\{u\}$\;
		\STATE $last\_seed = u$\;
		\STATE $Flag = 1$\;
		\STATE $continue$\;
		\ELSIF{$prev\_best[u] == last\_seed$}
		
		\STATE $mg1[u] \leftarrow mg2[u]$\;
		\ELSE
		
		\STATE $mg1[u]\leftarrow  \sigma_{\mathcal{T}}(\{u\}) \frac{\underset{(uw) \in E(\mathcal{G})}{\sum}\mathcal{P}_{u \rightarrow w} (1-\mathcal{P}_{\mathcal{S} \rightarrow w})\sigma_{\mathcal{T}}(\{w\})}{\underset{(uw) \in E(\mathcal{G})}{\sum}\mathcal{P}_{u \rightarrow w}\sigma_{\mathcal{T}}(\{w\})}$\;
		\STATE $prev\_best[u] = cur\_best$\;
		\STATE $\mathcal{M} \leftarrow \mathcal{S} \cup \{cur\_best\}$\;
		\STATE $mg2[u] \leftarrow \sigma_{\mathcal{T}}(\{u\}) \frac{\underset{(uw) \in E(\mathcal{G})}{\sum}\mathcal{P}_{u \rightarrow w} (1-\mathcal{P}_{\mathcal{M} \rightarrow w})\sigma_{\mathcal{T}}(\{w\})}{\underset{(uw) \in E(\mathcal{G})}{\sum}\mathcal{P}_{u \rightarrow w}\sigma_{\mathcal{T}}(\{w\})}$\; 
		\ENDIF
		\STATE $flag[u] = \mid \mathcal{S} \mid $\;
		\STATE $cur\_best \longleftarrow \underset{v \in V(\mathcal{G})/\mathcal{S},\mathcal{C}(\mathcal{S}\cup\{v\})\le\mathcal{B}}{argmax}mg1[v]$\;
		\ENDWHILE
		\STATE $return \ \mathcal{S}$\;
	\end{algorithmic}
\end{algorithm}
\par Algorithm \ref{Algo:4} gives a method that can even decrease the computational time of Algorithm \ref{Algo:3} and thus, make our proposed algorithm even more efficient. For each node $u$, we store the following information: (i) $mg1[u]$, which is the marginal gain of $u$ with respect to the current seed set, $\mathcal{S}$, (ii) $prev\_best[u]$, which is the node with the maximum marginal gain among all nodes considered in the present iteration before $u$, (iii) $mg2[u]$, which is the marginal gain of u with respect to the seed set {$\mathcal{S} \cup prev\_best$} and (iv) $flag[u]$ is the iteration number, when the $mg1[u]$ was last updated. We select each seed node, as we had done in Algorithm \ref{Algo:3}, choosing the node with the maximum $mg1$, which is the marginal influence gain for the node and add it to the seed set, if its $flag$ value is equal to the length of the seed set, which signifies its influence over the entire seed set. Here also, we do not need to calculate the value of $mg1$ for all the non-seed nodes in each iteration due to the sub\mbox{-}modularity property. However, optimization is done in calculating the value of $mg1$. If, for any iteration, the $prev\_best$ for the selected node, $u$ is the last seed node selected, then, we do not need to calculate the value of $mg1[u]$ all over again and can directly assign the value of $mg2[u]$ to $mg1[u]$. This is due to the fact, the value of $mg2[u]$ has already effectively been calculated with respect to the last added seed node. This algorithm has yielded a considerable decrease in the computational time, when tested with real\mbox{-}life datasets and thus, it proves to be more efficient than our previous algorithms. Here, we want to highlight that the computational complexity of both Algorithm \ref{Algo:3} and \ref{Algo:4} is same as Algorithm \ref{Algo:2} in the worst case. Hence, we do not calculate it separately.
\section{Experiments} \label{Sec:EE}
Here, the experimental details of the solution methodologies have been described. First, we briefly describe the datasets.
\subsection{Dataset Descrption}
The experimentation carried out in this work uses the following three real\mbox{-}life, publicly available social network datasets. 
\begin{itemize}
	\item \textbf{Email-Eu-core network} \footnote{\url{http://snap.stanford.edu/data/email-Eu-core.html}} \cite{leskovec2007graph}: This is a network among a group of persons and generated based on an  email exchange data from a European research institution. Between two persons $u$ and $v$ there is an edge, if $u$ sends a mail to person  $v$. 
	\item \textbf{Facebook Network} \footnote{\url{http://snap.stanford.edu/data/egonets-Facebook.html}} \cite{leskovec2012learning}: This is a Facebook ego network, where the nodes of the network are the users and an edge between two users signifies that the corresponding users are friends in Facebook. This dataset was created by a survey participants.
	\item \textbf{PHY Network} \footnote{\url{https://arxiv.org/archive/physics}} \cite{chen2010scalable} \cite{chen2009efficient}: This is an academic collaboration network among the researchers of the Physics section crawled from \url{Arxiv.org}. Each author is represented by a node and two nodes are connected, if the corresponding researchers coauthored at least one paper.
\end{itemize}

We download the first two datasets from \emph{Stanford Large Network Dataset Collection}\footnote{\url{http://snap.stanford.edu/data/index.html}} and the third one from \url{https://www.microsoft.com/en-us/research/people/weic/#!selected-projects}. The nature of the first two datasets are same, because they are formed based on the information exchanged among a group of users. However, the third one is an academic collaboration network, implicitly derived from the co\mbox{-}author relationships based on their submitted research articles. All these datasets have been extensively used in influence maximization literature \cite{swetha2017identification} \cite{wang2012scalable}. Table \ref{Tab:1} gives the basic statistics of the datasets. 
\begin{center}
	\begin{table}[H]
		\centering
		\caption{Basic statistics of the datasets.}
		\begin{tabular}{|l|l|l|l|l|}
			\hline
			\textit{Dataset Name} & $n$ & $m$ & \textit{Avg Deg} & \textit{Avg Clus Coeff} \\
			\hline
			Email-Eu-core network & 1005 & 25571 & 25.443 & 0.3994 \\
			\hline
			Facebook Dataset & 4039 & 88234 & 43.6910 & 0.6055 \\
			\hline
			PHY Network & 37154 & 231584 & 12.466 & 0.2371\\
			\hline
		\end{tabular}
		\label{Tab:1}
	\end{table}
\end{center}
\subsection{Experimental Setup}
\subsubsection{Setting of diffusion probability} In this study, we consider the following two diffusin probability, which are common in the literature \cite{arora2017debunking, tang2018efficient}.
\begin{itemize}
\item The first one is the \textit{uniform setting}, where each edge of the network has the same diffusion probability. Now, it is an important question what numerical value we should choose for this fixed probability value. Now, this is an context dependent issue and based on existing literature, in this paper we consider this value as $0.1$ (denoted as $p_c=0.1$) \cite{banerjee2019combim, xu2016finding}. 
\item the second one is the \textit{trivalency setting}, where edges have been assigned the probability value uniformly at random from the set $\{0.1, 0.01, 0.001\}$. This kind of setup is consistent with the existing literature \cite{nguyen2013budgeted}.
\end{itemize} 
\subsubsection{Setting for budget and cost} We choose the integers uniformly at random from the interval $[50,100]$ for assigning the node to its selection cost. For the budget, we initially start with the budget value of $2000$ and sequentially, add $2000$ to each iteration and continued till $16000$. This kind of setting has been previously adopted to study the BIM problem \cite{nguyen2013budgeted}.
\subsubsection{Setting for delay distribution and diffusion time} 
For assigning the delay in successive rounds, we use the \textit{Poisson Distribution}. For each $u_i$, the parameter for the distribution is randomly selected from the set $\{1, 2, \dots, 20\}$. We consider the maximum diffusion time as $10$ units. This setting has been adopted in \cite{liu2014influence} to study the influence maximization problem with delay.
\subsection{Algorithms in the Experimentation} \label{Algo:Comp}
For the clarity of understanding, here, we briefly mention the algorithms that are their in the experimentation.
\subsubsection{Algorithms Proposed in this Paper}
\begin{itemize}
	\item \textbf{Algorithm \ref{Algo:2}}: This is the incremental greedy procedure, which computes the marginal influence gain as stated in Theorem \ref{Th:1}.
	\item \textbf{Algorithm \ref{Algo:3}a}: This is basically the improvisation of Algorithm \ref{Algo:2} by exploiting the sub\mbox{-}modularity property of the influence function. Here, the marginal gain is computed using the normal influence estimation process.
	\item \textbf{Algorithm \ref{Algo:3}b}: This is a modified version of the  Algorithm \ref{Algo:3}, where the marginal gain is computed applying Theorem \ref{Th:1}.
	\item \textbf{Algorithm \ref{Algo:4}}: This improves Algorithm \ref{Algo:3} by exploiting the sub\mbox{-}modularity property of the time delayed influence function.
\end{itemize}
Here we highlight that we do not include Algorithm \ref{Algo:Intutive} in our experiments as it is highly inefficient.

\subsubsection{Baseline Algorithms}
We compare the performance of our methodology in influence spread with four baseline methods from the literature. Here, we give a very brief introduction to the methods.
\begin{itemize}
	\item \textbf{Maximum Degree Heuristic (DEG)}: This method iteratively selects the high degree nodes within the budget. In many previous studies, this method has been used as a baseline method, such as \cite{goyal2011simpath}, \cite{mohammadi2015time}.
	\item \textbf{Degree Discount Heuristic (DDH)}: This is a popular heuristic for the SIM problem proposed by \cite{chen2009efficient}. In this heuristic, if $u$ is a seed node and $(uv) \in E(\mathcal{G})$, then the degree of $v$ will be discounted by $2t_{v}+(d_{v}-t_{v})t_{v}\mathcal{P}_{u \rightarrow v}$, where $t_{v}$ is the number of neighbors of $v$ currently in the seed set, $d_{v}$ is the degree of $v$. This method has been used in  previous studies \cite{jiang2011simulated}.
	\item \textbf{Single Discount Heuristic (SDH)}: This a variant of degree discount heuristic proposed by \cite{chen2009efficient}. In this heuristic, if $u$ is a seed node and $(uv) \in E(\mathcal{G})$, then the degree of $v$ will be discounted by $1$. This method has been used as a baseline method in many previous studies \cite{cao2011maximizing} \cite{jiang2011simulated}.
	\item \textbf{Influence Ranking and Influence Estimation (IRIE)}: This is a popular heuristic for solving the influence maximization problem under IC Model proposed by \cite{jung2012irie}. This method has been used as a baseline in many previous studies on influence maximization \cite{lee2015query} \cite{lee2014fast}. 
\end{itemize}
All the algorithms are implemented on \textit{Python 3.4} environment along with \textit{NetworkX 1.9.1} package in a system with $5$ nodes and each node has $32$ cores and $64$ GB of RAM running by a Centos $6.7$ environment. For computing the influence spread due to the seed sets selected by different Algorithms, we consider the value of $\mathcal{R}$ as $10000$.
\subsection{Experimental Results}
The main goal of the experimental study is to compare the performance of the methodologies (both proposed and baseline) briefed in Section \ref{Algo:Comp}. Here, our focus is on two aspects. The main performance measure is the quality of the seed set selected by an algorithm, and this is measured by the expected number of influenced nodes. The secondary performance measure is its efficiency, i.e., the amount of computational time required for locating the seed nodes in the network. First, we report the performance based on influence spread and next, we describe the computational time required for seed set selection.
\subsubsection{Performance on Influence Spread}
Figure \ref{Influ:Email} shows the expected influence spread due to seed sets selected by different algorithms for the \textit{Email-Eu-core network}. From  Figure \ref{Influ:Email}a, it is clearly observed that there is a significant gap in the expected influence spread between the baseline methods and that proposed in this paper. As an example for $\mathcal{B}=16000$, among the baseline methods, IRIE has the highest spread, which is $871$. On the other hand, with the same budget Algorithm \ref{Algo:2} can achieve the spread of $977$. Hence, there is a gap of $10.54 \%$ with respect to the number of nodes of the network.  This is due to the following reason. In the proposed algorithms, the seed set is selected in an incrementally greedy manner, where in each iteration, the node causes the maximum marginal gain to be put into the seed set and hence, at the termination, these algorithms return one of the best quality seed sets with respect to the influence spread ability. On the other hand, the maximum degree heuristic (DEG) and its other two counter parts, namely single discount and degree discount heuristics are basically centrality\mbox{-}based methods and hence, there is a high chance that many of the highly central nodes are clustered in a localized zone and there is a significant overlap between the influence zones of two or more seed nodes, and this cause these methods less effective. However, as single discount and degree discount heuristics put some restrictions to avoid two adjacent nodes in seed set, hence, its performance is found to be quite better than that of the maximum degree heuristic. On the other hand IRIE selects seed nodes based on influence rank calculation. Hence, it generates more number of influenced nodes compared the centrality based heuristics.
\begin{figure}[!ht]
	\centering
	\begin{tabular}{cc}
		\includegraphics[scale=0.34]{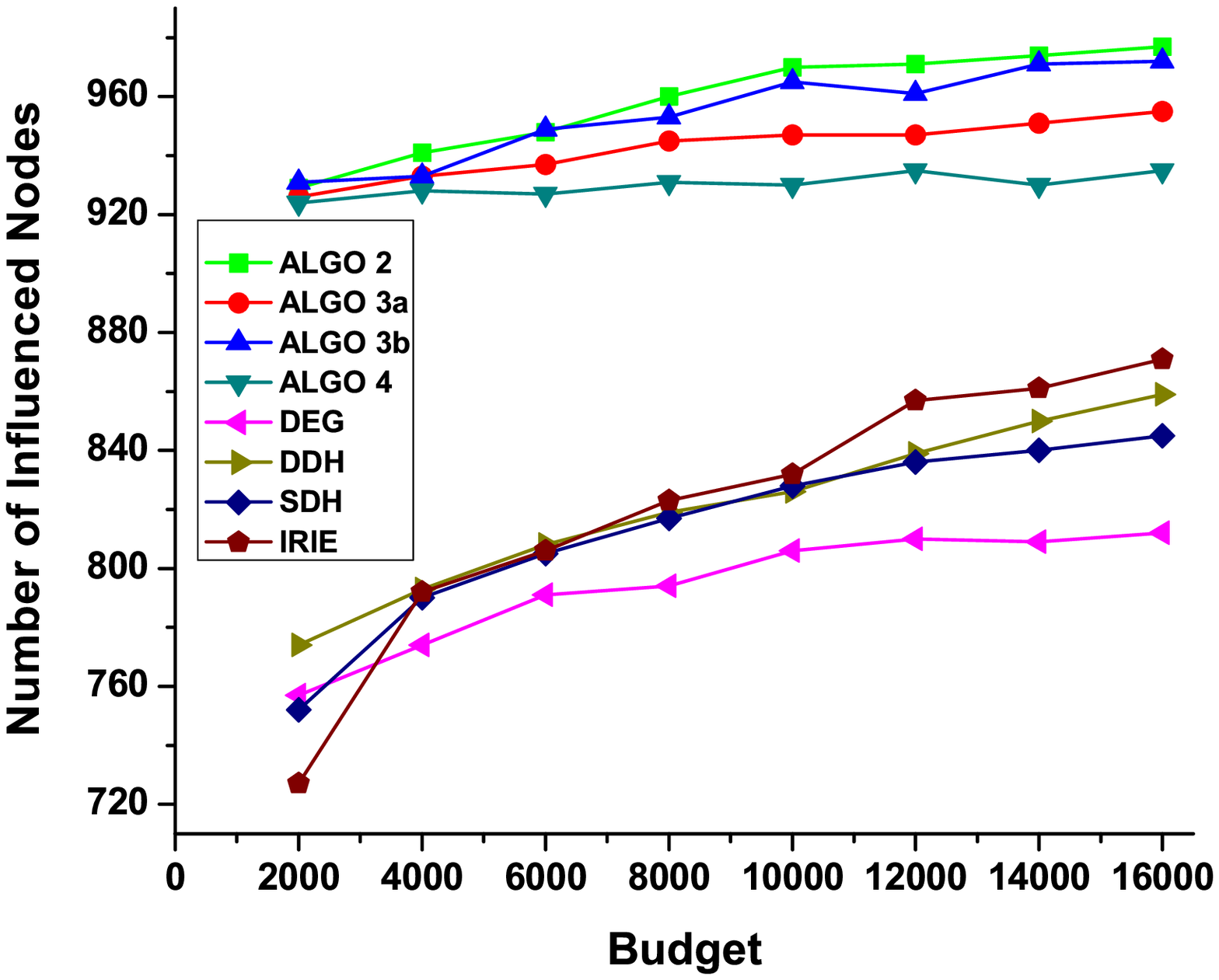} & \includegraphics[scale=0.34]{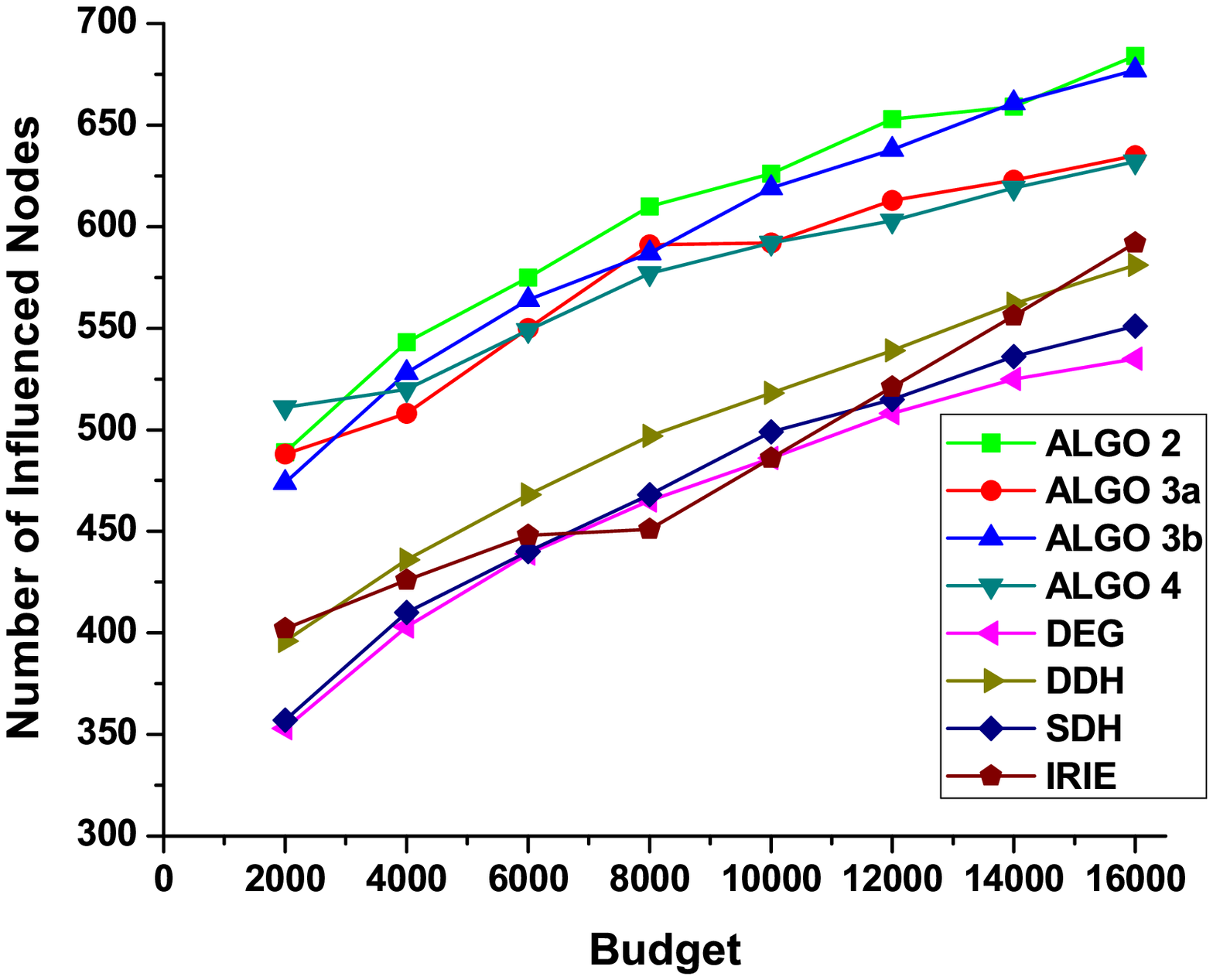} \\
		(a) Uniform (with $p_{c}=0.1$)  & (b) Trivalancy\\
	\end{tabular}
	\caption{Expected Number of Influenced Nodes for the \textbf{Email-Eu-Core Dataset} in Uniform and Trivalancy Settings due to the Seed Set Selected by Different Algorithms.}
	\label{Influ:Email}
\end{figure}
\par From the Figure \ref{Influ:Email}b, it is observed, that in tri\mbox{-}valency setting also, the seed selected by the proposed methods leads to more number of influenced nodes compared to the baseline methods. As an example, for $\mathcal{B}=16000$, the maximum number of influenced nodes due to Algorithm \ref{Algo:4} are 511, which is $50.85\%$ of the total nodes present in the network and also, this is found to be $44\%$, $29\%$, $43\%$, and $27\%$ more compared to DEG, DDH, SDH, and IRIE, respectively.
\par Next, we investigate the performance on influence spread for the Facebook Network dataset. Figure \ref{Influ:Facebook} shows the number of influenced nodes for different budget values on Facebook dataset. In this case also, we observe that, seed sets selected by the proposed methodologies lead to more number of influenced nodes compared to that of the baseline methods. As an example, for $\mathcal{B}=16000$, in uniform settings with $p_{c}=0.1$, the number of influenced nodes is $1782$, which is approximately $44\%$ of the total number of nodes and almost double compared to the number of influenced nodes due to the seed set selected by the IRIE Algorithm. In trivalancy setting, the number of influenced nodes due to the seed set chosen by Algorithm \ref{Algo:3}a is $923$, which is almost $23\%$ of the number of nodes and almost $55\%$ more compared to that of IRIE. 
\begin{figure}[!ht]
	\centering
	\begin{tabular}{cc}
		\includegraphics[scale=0.34]{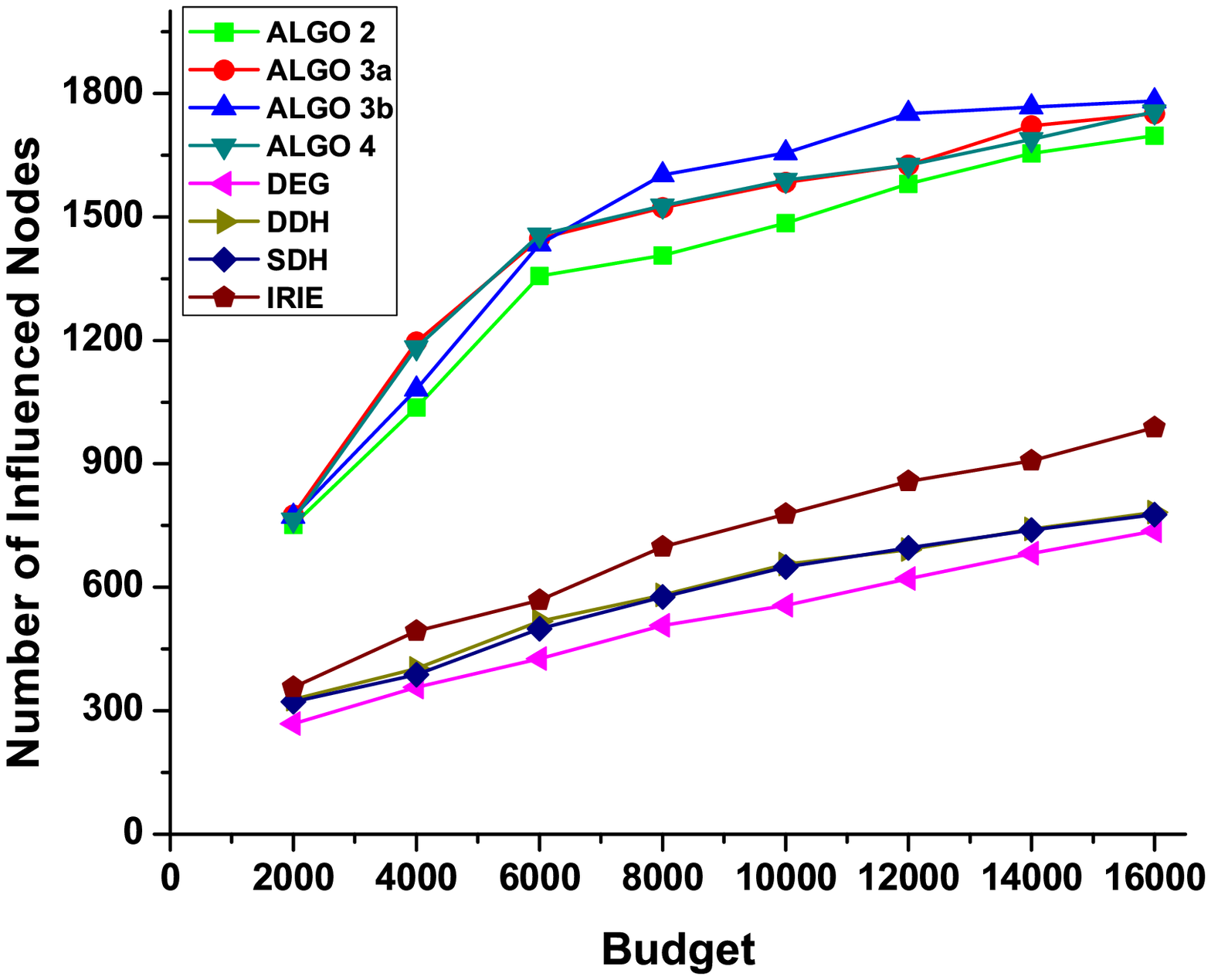} & \includegraphics[scale=0.34]{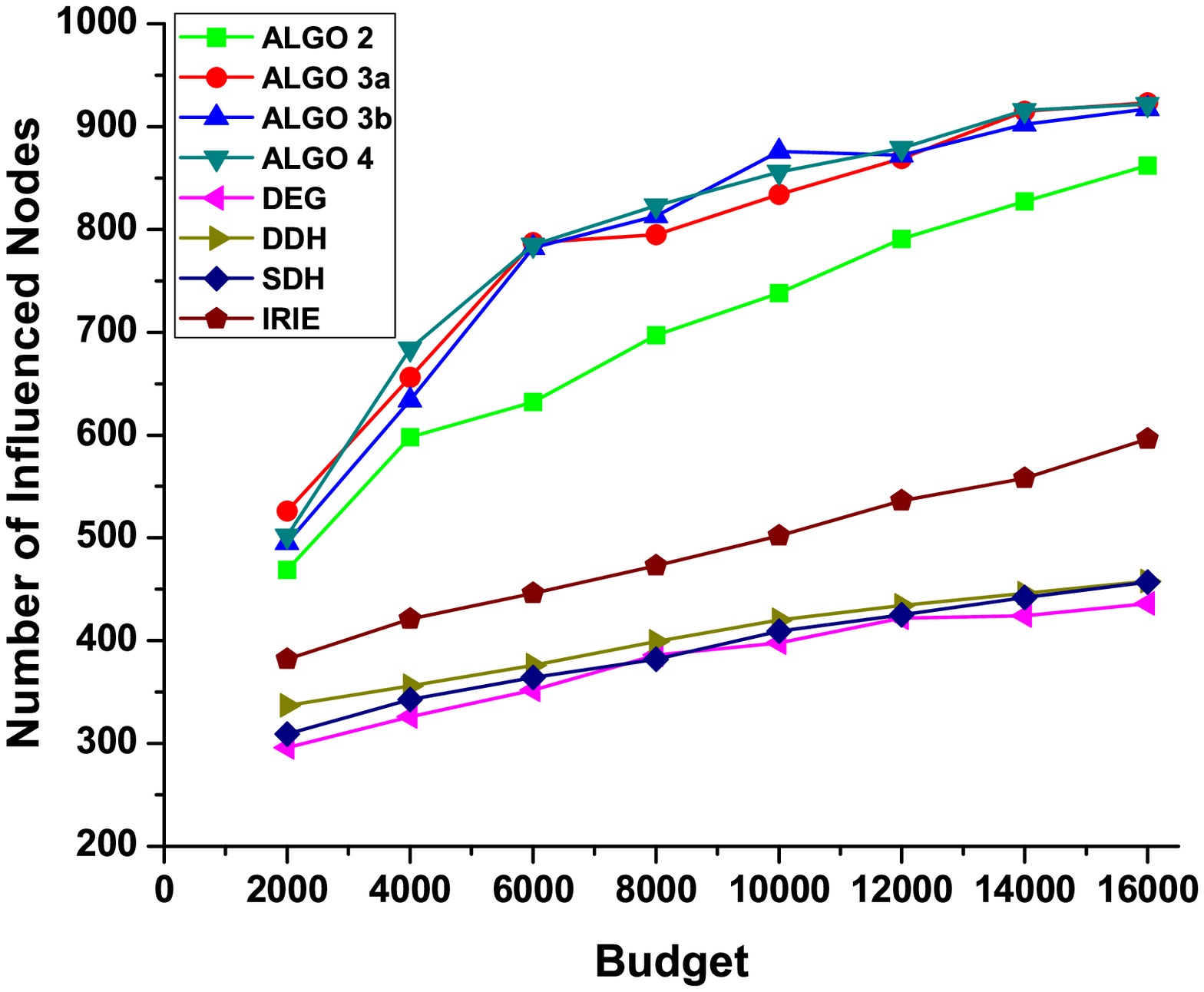} \\
		(a) Uniform (with $p_{c}=0.1$)  & (b) Trivalancy\\
	\end{tabular}
	\caption{Expected Number of Influenced Nodes for the \textbf{Facebook Dataset} in Uniform and Trivalancy Settings due to the Seed Set Selected by Different Algorithms.}
	\label{Influ:Facebook}
\end{figure}
\par Lastly, we investigate the performance of different algorithms on influence spread on Physics Collaboration Network dataset. Figure \ref{Influ:Phy} presents the influence spread due to the seed set selected by different algorithms for the  different budgets. From the Figure \ref{Influ:Phy}a, except for the lower budgets ($2000$ and $4000$), in this dataset also, we observe a significant difference in the expected influence spread by baseline methods and that proposed in this paper. As an example, for $\mathcal{B}=16,000$ among the baseline methods, the highest spread is achieved by IRIE and the number of influenced nodes in this case is found to be equal to $7446$ and among the proposed methods, the highest spread is achieved by Algorithm \ref{Algo:3}a leading to $11657$ number of influenced nodes. In comparison, this quantity is $50\%$ more compared to that of the baseline method. On the other hand, for the trivalancy model, the gap is only observed for higher budgets (greater than $8000$) and also seen to be significantly low compared to uniform setting. As an example, for $\mathcal{B}=16000$, the number of influenced nodes by the best baseline method, IRIE is found to be $2435$ and number of influenced nodes by Algorithm \ref{Algo:4} is $2776$, which is $14\%$ more. One possible explanation for this fact can be, for this dataset, the high degree nodes may be quite uniformly spreaded across the network. When the budget is low, naturally the less number of seed set can be selected. As the number of seed nodes is less, they are uniformly spreaded and overlapping zones of the seed nodes are also less. Contrary, when the budget is high, the number of selected seed nodes is also more in number. There is a possibility that some seed nodes may clustered into a particular zone of the network and this causes a significant overlap in the influence zone of these seed nodes. Hence, for the  higher budgets, the performance of DEG (consequently DDH and SDH) method on influence spread is quite poor.
\begin{figure}[!ht]
	\centering
	\begin{tabular}{cc}
		\includegraphics[scale=0.34]{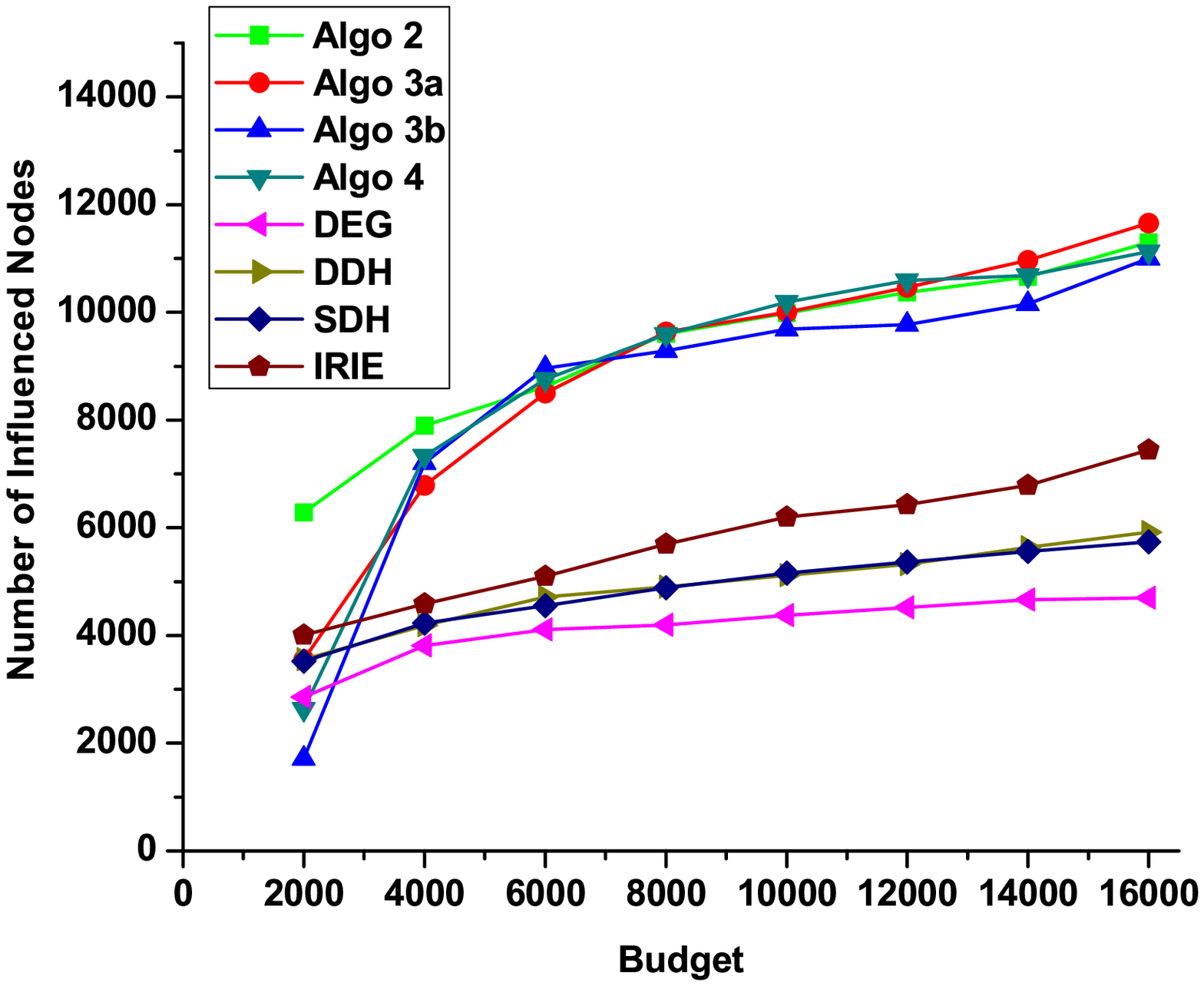} & \includegraphics[scale=0.34]{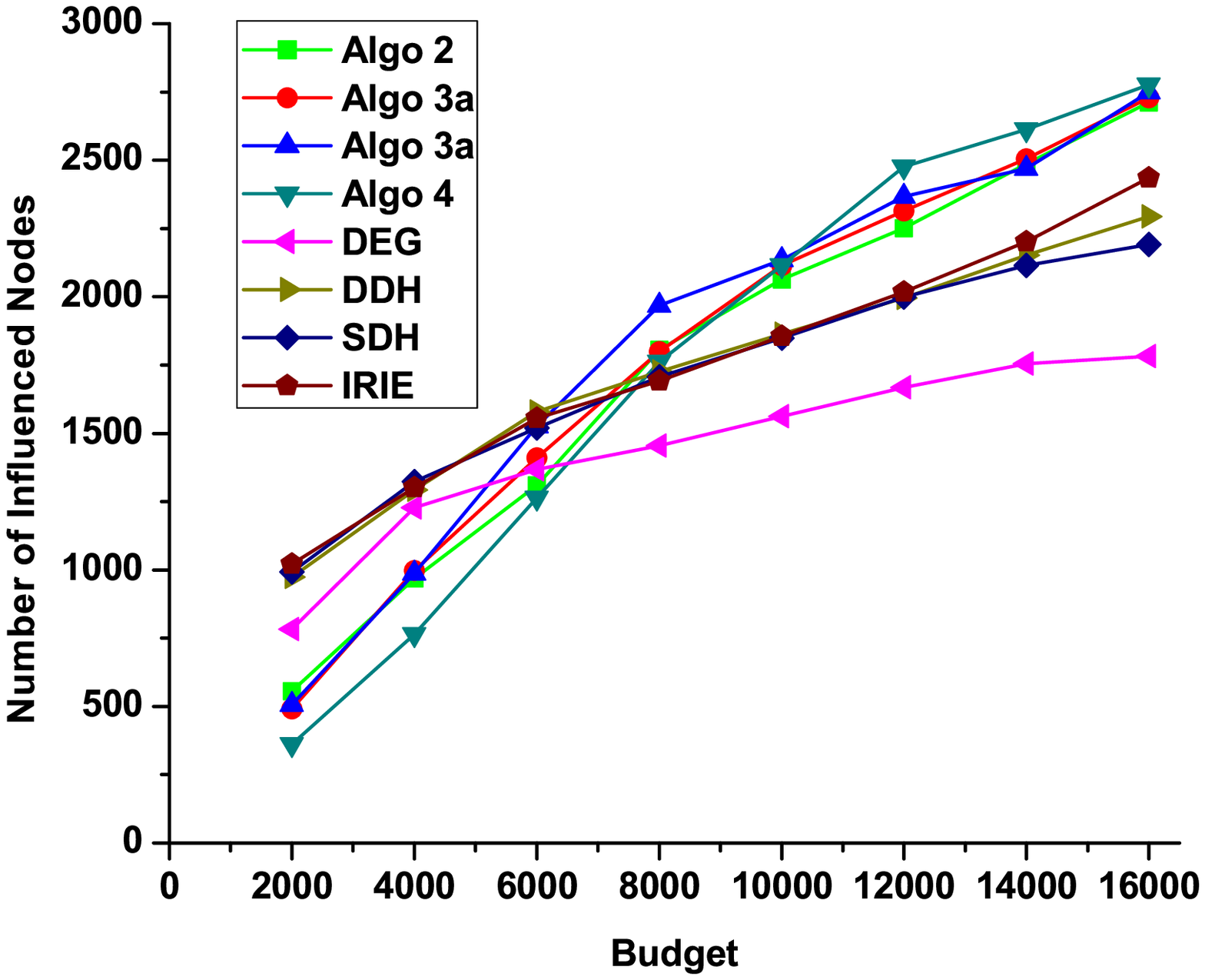} \\
		(a) Uniform (with $p_{c}=0.1$)  & (b) Trivalancy\\
	\end{tabular}
	\caption{Expected Number of Influenced Nodes for the \textbf{Physics Collaboration Network Dataset} in Uniform and Trivalancy Settings due to the Seed Set Selected by Different Algorithms.}
	\label{Influ:Phy}
\end{figure}
\subsubsection{Performance on Computational Time}
Here, we report the computational time for seed set selection. One point worthwhile to mention is that for the DEG, DDH and SDH heuristics the seed selection time is not dependent on the underlying diffusion probability setting. On the other hand, as the proposed methodologies and the IRIE compute influence spread during seed set selection, hence, underlying diffusion model accounts in computational time required for selecting the seed set. Therefor, for the proposed methodologies and the IRIE, we report the computational time, which is found to be maximum of uniform (with $p_{c}=0.1$) and trivalancy. Table \ref{Tab:Time_Email} reports the time requirement by different algorithms for finding the seed sets in \textit{Email-Eu-core network} dataset. As the DEG method just computes the degree and returns the high degree nodes, this is the fastest one. Along with computing the degree, DDH and SDH methods need to perform extra operations and therefore, these two methods take a little more time compared to the DEG method. However, the IRIE method is taking quite a long time because, during the rank calculation, a recursive procedure is followed, which is consuming a lot of time. 
\begin{center}
	\begin{table}[H]
		\caption{Computational time requirement (in Secs.) for seed set selection for \textbf{Email-Eu-core network} dataset}
		\begin{tabular}{  c  c  c  c  c  c  c  c  c }
			\hline
			Budget & 2000 & 4000  & 6000 & 8000 & 10000 & 12000 & 14000 & 16000 \\ \hline
			Algo \ref{Algo:2}  & 21.51 & 77.95 & 173.35  & 303.85 & 459.81 & 676.03 & 921.84 & 1184.73 \\ \hline
			Algo \ref{Algo:3}a & 1.86  & 2.15  & 2.53 & 2.96  & 3.47 & 4.03 & 4.66 & 5.29 \\  \hline
			Algo \ref{Algo:3}b & 25.68  & 90.19  & 186.20  & 305.54  & 464.64 & 635.29 & 859.56 & 1114.04 \\  \hline
			Algo \ref{Algo:4}  & 14.09  & 12.80  & 11.57  & 15.39  & 17.81 & 17.12 & 19.25 & 18.40 \\  \hline
			DEG   & 0.15  & 0.15   & 0.17  & 0.20  & 0.20 & 0.23 & 0.25 & 0.32 \\  \hline
			DDH   & 0.09  & 0.12   & 0.19  & 0.31  & 0.31 & 0.37 & 0.42 & 0.45 \\  \hline
			SDH   & 0.09  & 0.15   & 0.21  & 0.33  & 0.26 & 0.28 & 0.39 & 0.45 \\  \hline
			IRIE   & 62.68  & 138.46   & 226.31  & 392.67  & 426.12 & 459.51 & 502.12 & 535.15  \\  
			\hline
		\end{tabular}
		\label{Tab:Time_Email}
	\end{table}
\end{center} 
From this result, it is also clarified that the proposed approach developed based on the approximate marginal gain computation (Algorithm \ref{Algo:2}) is faster than IRIE for the lower budget values. With the increase of budget, this approach takes more time compared to rest of the methodologies. However, Algorithm \ref{Algo:3}a drastically improves the computational time of Algorithm \ref{Algo:2} by removing a lot of unnecessary computations. This observation is consistent with the intuition, as mentioned previously. As an example, when $\mathcal{B}=16000$, Algorithm \ref{Algo:3}a is more than $220$ times faster than Algorithm \ref{Algo:2}. In case of Algorithm \ref{Algo:3}b, as the gain is computed using simulation\mbox{-}based approach, it takes more time compared to Algorithm \ref{Algo:3}a. Algorithm \ref{Algo:4} tries to improve over the Algorithm \ref{Algo:3}a. However, this is not reflected in the result due to the following reason. To take the advantage of Algorithm \ref{Algo:4}, initially we need to evaluate the influence function in Lines $9$ and $11$, however, in Algorithms \ref{Algo:3}a and \ref{Algo:3}b, it is done once in Line 2. As the size of the \textit{Email-Eu-core network} dataset is small (consisting of $1005$ nodes and $25571$ edges), the improvement in time is not detected. Algorithm \ref{Algo:4} improves the computational time of Algorithm \ref{Algo:2} in a significant way. As an example, for $\mathcal{B}=16000$, Algorithm \ref{Algo:4} is almost $65$ times faster than Algorithm \ref{Algo:2}. Also, this observation is consistent, as stated previously.  
\par Now, we report computational time required for Facebook dataset in Table \ref{Tab:Facebook}. In this dataset also, the computational time requirement of different algorithms are found to be consistent with our intuitions. Particularly, for the higher budgets, Algorithm \ref{Algo:4} takes less amount of computational time compared to Algorithm \ref{Algo:3}a.
\begin{center}
	\begin{table}[H]
		
		\caption{Computational time requirement (in Secs.) for seed set selection for \textbf{Facebook} dataset}
		
		\begin{tabular}{  c  c  c  c  c  c  c  c  c }
			\hline
			Budget & 2000 & 4000  & 6000 & 8000 & 10000 & 12000 & 14000 & 16000 \\ \hline
			Algo \ref{Algo:2}  & 232.51 & 497.15 & 659.17  & 936.87 & 1327.71 & 1596.42 & 1846.63 & 1996.24 \\ \hline
			Algo \ref{Algo:3}a & 14.21  & 18.95  & 21.89 & 28.95  & 31.23 & 30.76 & 35.51 & 41. 98 \\  \hline
			Algo \ref{Algo:3}b & 198.21  & 256.89  & 392.67  & 478.34  & 535.87 & 566.45 & 598.66 & 634.04 \\  \hline
			Algo \ref{Algo:4}  & 22.21  & 22.01  & 23.56  & 24.92  & 26.78 & 29.82 & 33.81 & 33.26 \\  \hline
			DEG   & 2.16  & 2.17   & 2.59  & 2.48  & 3.01 & 2.93 & 3.15 & 3.12 \\  \hline
			DDH   & 2.52  & 2.69  & 2.93  & 2.67  & 3.01 & 3.17 & 3.42 & 3.39 \\  \hline
			SDH   & 2.51  & 2.57   & 2.52  & 2.59  & 2.62 & 2.81 & 2.79 & 2.91 \\  \hline
			IRIE   & 180.26  & 252.64  & 335.31  & 446.29  & 502.21 & 595.14 & 656.12 & 702.15  \\  
			\hline
		\end{tabular}
		\label{Tab:Facebook}
	\end{table}
\end{center} 
Now, we report the required computational time for seed set selection on Physics Collaboration Network dataset in Table \ref{Tab:Physics}. In this dataset, for the lower budget values, the efficiency of Algorithm \ref{Algo:3}a is not very significant. However, when the budget value becomes more than $10,000$, Algorithm \ref{Algo:3}a takes less time compared to Algorithm \ref{Algo:2}. One important fact to observe is that as this dataset is quite large, Algorithm \ref{Algo:4} is taking less time compared to Algorithm \ref{Algo:3}a. Hence, the efficiency of Algorithm \ref{Algo:4} is noticeable here. It can be also be observed that for the Algorithms \ref{Algo:3}a, \ref{Algo:3}b and \ref{Algo:4}, computational time remains almost the same for the varying budget values.
\begin{center}
	\begin{table}[H]
		\caption{Computational time requirement (in Secs.) for seed set selection for \textbf{Physics Collaboration Network} dataset}
		\begin{tabular}{  c  c  c  c  c  c  c  c  c }
			\hline
			Budget & 2000 & 4000  & 6000 & 8000 & 10000 & 12000 & 14000 & 16000 \\ \hline
			Algo \ref{Algo:2}  & 2146.33 & 2962.28 & 3752.72  & 5382.98 & 8540.09 & 14489.40 & 14594.12 & 17934.91 \\ \hline
			Algo \ref{Algo:3}a & 5806.96  & 5800.62  & 5608.59 & 5674.62  & 6171.75 & 6193.59 & 6081.62 & 5769.38 \\  \hline
			Algo \ref{Algo:3}b & 6206.69  & 6512.43  & 6526.39  & 6335.41  & 6221.57 & 6325.21 & 6317.92 & 6301.04 \\  \hline
			Algo \ref{Algo:4}  & 5193.91  & 5199.29  & 5206.20  & 5215.30  & 5225.46 & 5237.41 & 5251.82 & 5267.78 \\  \hline
			DEG   & 1.81  & 4.02   & 5.00  & 5.89  & 6.95 & 6.73 & 11.67 & 10.56 \\  \hline
			DDH   & 2.26  & 13.11  & 13.57  & 11.86  & 9.97 & 9.47 & 4.82 & 5.53 \\  \hline
			SDH   & 1.97  & 11.67   & 3.29  & 10.15  & 2.02 & 10.89 & 1.52 & 1.60 \\  \hline
			IRIE   & 62.68  & 151.46   & 401.51  & 562.67  & 791.12 & 906.51 & 1056.12 & 1236.15  \\  
			\hline
		\end{tabular}
		\label{Tab:Physics}
	\end{table}
\end{center} 
The important points to mention here is that, in real information diffusion scenarios, the most important thing from the advertiser’s perspective is the number of influenced nodes. At this point, the proposed methodologies are far ahead compared to the baseline methods. Secondly, as the size of the dataset
increases, the scalability of the Algorithm \ref{Algo:3}a, \ref{Algo:3}b and \ref{Algo:4} with respect to the budget also increase. These show the effectiveness and efficiency of the proposed methodologies.
\section{Conclusions and Future Directions} \label{Sec:CFD}
In this paper, we have introduced the problem of Budgeted Influence Maximization with delay by considering intermediate delay and time\mbox{-}sensitive nature of many real\mbox{-}life diffusion processes. For this problem, we have proposed an incremental greedy methodology, which works based on the approximate marginal spread computation. Time and space requirement analysis of this method has been done and these are found to be  linear with respect to the number of nodes and edges of the network. To deal with the larger datasets, we have improved the efficiency of this method by exploiting the sub\mbox{-}modularity property of the time delayed influence function. Experimentation with real\mbox{-}world social network datasets demonstrates that the proposed methodologies will be able to choose the seed nodes within feasible computational time, that leads to more number of influenced nodes compared to the baseline methods. Now, this work can be extended in several directions. One immediate extension of this study is to consider the presence of competitor and particularly a game theoretic model will be interesting.

%
%

\bibliographystyle{spbasic}      
\bibliography{Paper}   


\end{document}